%% file: icdm23.tex
\documentclass[conference]{IEEEtran}
\IEEEoverridecommandlockouts
\pdfoutput=1
\usepackage{cite}
\usepackage{amsmath,amssymb,amsfonts}
\usepackage{graphicx}
\usepackage{textcomp}
\usepackage{xcolor}

\usepackage{color}
\usepackage{amsthm}
\usepackage{algorithm}
\usepackage[noend]{algpseudocode}
\usepackage{bm}
\usepackage{float}
\usepackage{enumitem}
\usepackage{url}

\def\BibTeX{{\rm B\kern-.05em{\sc i\kern-.025em b}\kern-.08em
    T\kern-.1667em\lower.7ex\hbox{E}\kern-.125emX}}
\begin{document}

\title{Rule Mining for Correcting Classification Models
}

\author{\IEEEauthorblockN{1\textsuperscript{st} Hirofumi Suzuki}
\IEEEauthorblockA{\textit{Fujitsu Limited} \\
suzuki-hirofumi@fujitsu.com}
\and
\IEEEauthorblockN{2\textsuperscript{nd} Hiroaki Iwashita}
\IEEEauthorblockA{\textit{Fujitsu Limited} \\
iwashita.hiroak@fujitsu.com}
\and
\IEEEauthorblockN{3\textsuperscript{rd} Takuya Takagi}
\IEEEauthorblockA{\textit{Fujitsu Limited} \\
takagi.takuya@fujitsu.com}
\and
\IEEEauthorblockN{4\textsuperscript{th} Yuta Fujishige}
\IEEEauthorblockA{\textit{Fujitsu Limited} \\
fujishige.yuta@fujitsu.com}
\and
\IEEEauthorblockN{5\textsuperscript{th} Satoshi Hara}
\IEEEauthorblockA{\textit{Osaka University} \\
satohara@ar.sanken.osaka-u.ac.jp}
}

\include{macro}

\maketitle

\begin{abstract}
Machine learning models need to be continually updated or corrected to ensure that the prediction accuracy remains consistently high.
In this study, we consider scenarios where developers should be careful to change the prediction results by the model correction, such as when the model is part of a complex system or software.
In such scenarios, the developers want to control the specification of the corrections.
To achieve this, the developers need to understand which subpopulations of the inputs get inaccurate predictions by the model.
Therefore, we propose \emph{correction rule mining} to acquire a comprehensive list of rules that describe inaccurate subpopulations and how to correct them.
We also develop an efficient correction rule mining algorithm that is a combination of frequent itemset mining and a unique pruning technique for correction rules.
We observed that the proposed algorithm found various rules which help to collect data insufficiently learned, directly correct model outputs, and analyze concept drift.
\end{abstract}


\section{Introduction}

Machine learning models are now ubiquitous in society, including attempts on critical tasks such as medical diagnosis~\cite{caruana2015intelligible}, credit scoring~\cite{siddiqi2012credit}, and predictive justice~\cite{kleinberg2018human}. 
In cases of using prediction models for such critical tasks, the models must be continually updated or \emph{corrected} to ensure that the prediction accuracy remains consistently high.
A standard model correction scenario is that there is a prediction model now in operation, and over time, accuracy degradation will be detected on newly collected data, then developers should make efforts to improve model performances based on the new data.
In fact, it is widely known that the model performance gets worse as concept drift occurs over time~\cite{gama2014survey}.
In addition, there are fully automatic ways for adapting models to new data, i.e., domain adaptation techniques~\cite{domain-adaptation-survey}.

In this study, we consider more critical model correction scenarios such that some consistencies are required in models before and after correction.
For example, when a model is part of a complex system or software and affects several processes, developers should carefully change the prediction results by model correction.
In such a scenario, it is desirable that developers can control what predictions are corrected instead of using fully automatic correction.
Then, developers need to understand the target correction regions and how to correct them.
In general, because the correction is to reduce misclassifications, our target regions can be represented by inaccurately predicted subpopulations.
Therefore, we aim to acquire knowledge of inaccurate subpopulations as a comprehensive list of rules describing the characteristics of subpopulations and how to correct them.

\begin{figure}[t]
\centering
\includegraphics[scale=0.65]{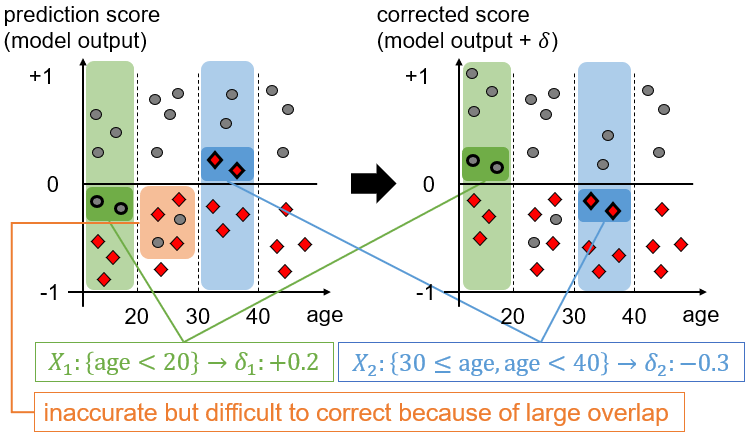}
\caption{
The concept of correction rules.
The left figure represents the prediction scores for instances of positive (black circle) and negative (red diamond) classes on the current model with some misclassifications.
The right figure represents the corrected prediction scores after applying the two correction rules $X_1 \to \delta_1$ and $X_2 \to \delta_2$ at the bottom of the figure.
With these correction rules, the misclassified instances of ``age$<$20'' and ``30$\leq$age$<$40'' get correct predictions by prediction scores with $\delta_1$ and $\delta_2$ are added, respectively.
Although there are some misclassifications with the ages between 20 and 30, we consider that it is difficult to correct because the scores of both classes highly overlap.
}
\label{fig:concept}
\vspace{-4mm}
\end{figure}

We particularly focus on correcting binary classification models, and we define {\it correction rules} that describe inaccurate subpopulations and how to correct them.
Figure~\ref{fig:concept} illustrates the concept of correction rules.
Correction rules attempt to change prediction scores output by a given base model.
The sign of an output score represents the predicted class for an input instance. 
For example, if the score is $s = 0.7$ or $s = -0.3$, the instance is predicted as positive or negative, respectively.
A correction rule is denoted by $X \rightarrow \delta$ where $X$ is a condition specifying subpopulation and $\delta$ is an optimized correction amount of model outputs.
More formally, $X$ is denoted by a set of discretized attributes such as ``age $<$ 30'', ``age $\geq$ 20'', and ``occupation = teacher''.
And, $\delta$ changes the score $s$ into $s + \delta$ for subpopulation fall within $X$.

Using the definition of correction rules, we propose a new mining problem called {\it correction rule mining}.
If $\delta$ greatly improves the classification accuracy on the subpopulation fall within $X$, we consider that the correction rule $X \rightarrow \delta$ is effective.
Even if a condition $X$ captures many misclassified instances, the rule is ineffective if no $\delta$ can improve accuracy.
Our correction rule mining aims to find only effective correction rules.
Therefore, we quantify the effectiveness of rules by introducing metrics \emph{support} and \emph{confidence} in correction rule mining: the frequency of misclassified but corrected instances and the accuracy of the correction on captured subpopulation.
We then formalize the correction rule mining as the problem of enumerating all the correction rules that exceed the thresholds.
We show that the problem can be solved by combining {\it frequent itemset mining}~\cite{Agrawal:apriori} and a unique pruning technique using partial monotonicity on the confidence metric.

\subsubsection*{Our Contributions}
We introduce a new problem, correction rule mining, of enumerating all pairs of an inaccurate condition and an optimized correction amount.
We then propose the first efficient mining algorithm for this problem. 
With the proposed algorithm, we can automatically obtain a comprehensive list of rules describing inaccurate subpopulations and how to correct them, which helps the developer's decision to improve training data, directly correct model outputs, and analyze concept drift.
In summary, our contributions are as follows:
\begin{enumerate}
    \item We define the concept of correction rules that help to improve model performances by focusing on inaccurate subpopulations, and formulate the task of enumerating correction rules for the first time.
    \item We propose an efficient correction rule mining algorithm that combines frequent itemset mining and a unique pruning technique for correction rules.
    \item The proposed algorithm is model agnostic, i.e., it can be applied to any binary classification model with a function to output prediction scores.
\end{enumerate}
Our experimental results confirm the validity of our enumeration approach compared to the existing rule-based model correction approaches in terms of accuracy improvements and concept drift summarization.

\section{Related Work}

\subsection{Rule Mining}
There are many studies on rule mining~\cite{Agrawal:apriori,Dong_KDD1999,Loekito_KDD2006,Novak_JMLR2009,SubgroupDiscovery} which attempt to enumerate conditions for subpopulations with high proportions of instances having a specified label.
Therefore, we can use them to obtain rules related to misclassified subpopulations by labeling misclassified or not.
However, we can not guarantee that the rules are suitable for correcting classification models because there are no validations of rules that are accurate for model correction.
In contrast, our correction rule mining can obtain rules for properly correcting classification models by considering correction amounts and the confidence of the correction.

There is another approach, called PREMISE~\cite{PREMISE}, of finding rules for misclassified subpopulations.
PREMISE penalizes rules by the minimum description length (MDL) principle and outputs a few rules as a summarization of misclassifications.
Therefore, PREMISE does not fit our purpose of finding a comprehensive list of rules (i.e., enumeration) without overlooking valuable conditions of misclassifications.

\subsection{Tree-Based Correction}
In machine learning, decision tree structures are used as popular models of rule mining.
LCT~\cite{LCT} is a decision tree to directly correct classification models by stacking it on the base models and optimizing linear correction of prediction scores.
Traversing LCT from root to leaf, we can obtain rules when the model tends to misclassify and how it can be corrected.
However, because LCT relies on the tree structure, the rules become mutually exclusive while sharing the rules in common ancestor nodes.
These restrictions limit the diversity of the rules we can find.
Therefore, LCT cannot find a comprehensive list of correction rules, and it does not fit our purpose.

\subsection{Patching Framework}
Patching \cite{Patching} is also a related study of model corrections.
The literature proposes a framework for correcting classification models by detecting misclassification regions and constructing patch models to re-predict instances that fall into the misclassification regions. 
However, the literature had not discussed deeper how to detect misclassification regions, and constructing a patch model would be more effective.
In fact, the experiments in the literature had tried only a tree-based approach for misclassification detection and patch models.
Therefore, we can say that LCT and our correction rules give the patching framework another concrete approach:
detecting misclassification regions by if-then rules and constructing a patch model of only linear or scalar calculations.

\section{Problem Formulation}

\subsection{Itemsets and Datasets}
Let $I = \{x_1, x_2, \ldots, x_n\}$ be a set of $n$ items.
We assume that each item is written as the form ``$x_i = \mathrm{attr}_i ~ \circ_i ~ \mathrm{value}_i$'' where $\mathrm{attr}_i$ indicates an attribute, $\mathrm{value}_i$ is a value of the attribute, and $\circ_i \in \{<, \geq, =, \neq\}$ is a relational operator between the attribute and the value.
For example, ``age $<$ 30'', ``age $\geq$ 25'', and ``occupation = teacher'' are items.
An itemset is a non-empty subset of items $X \subseteq I$, for example \{``age $<$ 30'', ``age $\geq$ 25'', ``occupation = teacher''\} that represents a logical expression ``25 $\leq$ age $<$ 30 $\land$ occupation = teacher''.

We implicitly handle machine learning models for binary classification tasks.
Let $1$ and $-1$ be the class labels of the positive and negative classes, respectively.
Now, we have a classification model of interest in our problem, which we call the base model.
We assume that the base model can output a prediction score $s \in (-1, 1)$ for each prediction\footnote{If the score is defined on $(-\infty, \infty)$, we can transform them to $(-1, 1)$ by using $\tanh$.} where the predicted class label is defined by
\begin{equation*}
\ell(s) :=
\begin{cases}
    1 & s > 0,\\
    -1 & s \leq 0.
\end{cases}
\end{equation*}
This setting is standard in machine learning.

In our problem, we are given a dataset as a multiset $D \subseteq 2^I \times (-1,1) \times \{-1,1\}$.
Each instance $(t, s, c) \in D$ is a tuple of an itemset $t \subseteq I$ associated with a prediction target, a prediction score $s$ of the base model, and a true class label $c$.
We can divide the dataset $D$ into four disjoint datasets according to the prediction results: true positive dataset $D_+^\top := \{(t, s, c) \in D \mid \ell(s) = c = 1\}$, false positive dataset $D_+^\bot := \{(t, s, c) \in D \mid \ell(s) = 1, c = -1\}$, true negative dataset $D_-^\top := \{(t, s, c) \in D \mid \ell(s) = c = -1\}$, and false negative dataset $D_-^\bot := \{(t, s, c) \in D \mid \ell(s) = -1, c = 1\}$.

\subsection{Correction Rules}
We consider correcting prediction results of the base model by {\it correction rules} each of which is defined as $X \rightarrow \delta$ where $X \subseteq I$ is an itemset and $\delta \in [-1, 0) \cup (0, 1]$ is a correction amount.
Given an instance of an itemset $t \subseteq I$ and a prediction score $s \in (-1,1)$, a correction rule $X \rightarrow \delta$ means that the prediction score is changed from $s$ to $s + \delta$ if the condition $X \subseteq t$ holds.
If $\delta > 0$ holds, then we say that the correction is in the positive direction; otherwise, the negative direction.

We define some metrics for estimating the quality of correction rules.
While traditional rule mining considers metrics focusing on all the hit instances, our metrics focus on the hit ones with their predictions changed.
Let $R$ be the short form of the correction rule $X \rightarrow \delta$.
We define the sets of truly changed and falsely changed instances $C_R^\top$ and $C_R^\bot$ as follows:
\begin{align*}
    C_R^\top := \{(t,s,c) \in D_+^\bot \cup D_-^\bot \mid X \subseteq t, \ell(s + \delta) = c\}
    \\
    C_R^\bot := \{(t,s,c) \in D_+^\top \cup D_-^\top \mid X \subseteq t, \ell(s + \delta) \neq c\}
\end{align*}
From the above definition, we also define the {\it support} and {\it confidence} in terms of correction as follows:
\begin{align*}
    \mathrm{supp}(R) := \frac{|C_R^\top|}{n^\bot},\quad \mathrm{conf}(R) := \frac{|C_R^\top|}{|C_R^\top| + |C_R^\bot|}
\end{align*}
where $n^\bot = |D^\bot_-|$ if $\delta > 0$, otherwise $n^\bot = |D^\bot_+|$.
They indicate the frequency of the truly changed predictions and the accuracy of the prediction changes, respectively.
Although the correction amount $\delta$ does not become zero by the definition of the correction rule, we denote $\mathrm{supp}(X \rightarrow 0) = \mathrm{conf}(X \rightarrow 0) = 0$ for convenience.

\subsection{Optimizing Correction Amount}
\label{subsec:optimize-correction-amount}
We define an optimized correction amount for any pair of an itemset and a direction.
Let $\theta$ be the given support threshold.
We optimize correction amounts for an itemset $X$ in each direction as follows:
\begin{align*}
    \delta^*_+(X) \in \argmax_{\delta \in (0,1]} \{\mathrm{conf}(X \rightarrow \delta) \mid \mathrm{supp}(X \rightarrow \delta) \geq \theta\}\\
    \delta^*_-(X) \in \argmax_{\delta \in [-1,0)} \{\mathrm{conf}(X \rightarrow \delta) \mid \mathrm{supp}(X \rightarrow \delta) \geq \theta\}
\end{align*}
If such correction amounts do not exist, we set $\delta^*_+(X) = 0$ or $\delta^*_-(X) = 0$ for convenience.
In practice, we can restrict the candidates of correction amount $\delta$ based on the given dataset.
Let us consider the sequence $-1 = p_0 < p_1 < \ldots < p_k < p_{k+1} = 1$ of distinct prediction scores on $(t, s, c) \in D$.
We can then find $\delta^*_+(X)$ and $\delta^*_-(X)$ in $\{-\frac{p_i + p_{i+1}}{2} \mid 0 \leq i \leq k\}$.

\subsection{Correction Rule Mining}
We now formulate {\it correction rule mining} using the notations above.
We consider that a correction rule is better if it has enough support and confidence.
Moreover, on the aspect of readability, it is better to limit the itemset length.
Concluding such requirements as the hard constraints, the formulation of correction rule mining is as follows:
\begin{dfn}[Correction Rule Mining]
Given a set $I$ of items, a dataset $D$, a support threshold $\theta \in [0,1]$, a confidence threshold $\lambda \in [0,1]$, and a maximum itemset length $K \in \mathbb{N}$, we say a correction rule $X \rightarrow \delta$ is acceptable if it satisfies all the following conditions:
(i) $|X| \leq K$, (ii) $\mathrm{supp}(X \rightarrow \delta) \geq \theta$, (iii) $\mathrm{conf}(X \rightarrow \delta) \geq \lambda$, and (iv) $\delta = \delta^*_+(X)$ or $\delta = \delta^*_-(X)$.
Correction rule mining is to enumerate all the acceptable correction rules.
\end{dfn}
%

\section{Algorithms}
We propose an algorithm for correction rule mining using a backtrack search.
However, a naive implementation of the backtrack search can be prohibitively slow because the number of correction rule candidates can be exponentially huge.
Therefore, we need some acceleration techniques using properties of correction rules.

First, we show a relationship between correction rule mining and {\it frequent itemset mining}, and propose to use an algorithm for frequent itemset mining as a subroutine of our algorithm.
Second, we show a monotonicity of confidence in partial search spaces and propose an efficient pruning strategy unique to correction rule mining.
Third, we propose an algorithmic option to enumerate only {\it minimal} correction rules to reduce the output size.
In addition, we describe post-processing to denoise enumerated correction rules for practical use.

The proposed algorithm runs on each pair of datasets $(D^\top_+, D^\bot_+)$ and $(D^\top_-, D^\bot_-)$, i.e., for each direction of correction.
In the following, we denote the input datasets of the algorithm by $D^\top$ and $D^\bot$ without loss of generality.

\subsection{Relationship with Frequent Itemsets}
We can find that a set of acceptable correction rules are derived from a subset of frequent itemsets according to a simple observation.
Given an itemset $X \subseteq I$ and a dataset $D$, let $D(X) := \{(t,s,c) \in D \mid X \subseteq t\}$ be the set of hit instances by $X$.
The frequent itemsets are defined as follows:
\begin{dfn}[Frequent Itemset]
Given a threshold $\theta' \in [0,1]$, an itemset $X \subseteq I$ is a frequent itemset if $\frac{|D(X)|}{|D|} \geq \theta'$ holds.
\end{dfn}
\noindent
We obviously have the equation $\mathrm{supp}(X \rightarrow \delta) \leq \frac{|D^\bot(X)|}{|D^\bot|}$ for any $\delta$.
Therefore, an itemset $X$ has a chance to become a correction rule if and only if $X$ is a frequent itemset on $D^\bot$ with the threshold $\theta' = \theta$.
In addition, we focus on a lattice structure of frequent itemsets defined as follows:
\begin{dfn}[Equivalent Lattice]
For any frequent itemsets $X, S \subseteq I$ satisfying $X \cap S = \emptyset$ and $D(X) = D(X \cup S)$, we call $L(X,S) := \{Y \subseteq I \mid X \subseteq Y \subseteq X \cup S\}$ equivalent lattice.
For any $Y \in L(X,S)$, we have $D(X) = D(Y)$.
\end{dfn}
%
Using equivalent lattices, we can enumerate frequent itemsets as a compressed form so that each equivalent lattice represents multiple frequent itemsets.
In fact, the algorithm named LCMfreq in the literature~\cite{LCM} realizes the idea above.
LCMfreq outputs frequent itemsets by enumerating a set of equivalent lattices $\mathcal{L}$.
Here, $L_1 \cap L_2 = \emptyset$ for any two distinct lattices $L_1, L_2 \in \mathcal{L}$ and $\bigcup_{L \in \mathcal{L}} L$ become the set of all the frequent itemsets.
Therefore, we can complete correction rule mining by scanning each lattice $L(X,S) \in \mathcal{L}$ of $|X| \leq K$ obtained by LCMfreq on the dataset $D^\bot$ and the threshold $\theta$.

\subsection{Monotonicity on Equivalent Lattice}
We start to scan an equivalent lattice $L(X,S)$ by a backtracking strategy starting with $X$.
We recursively generate new itemsets until the length limit is met, optimize their correction amounts, and check whether each correction rule is acceptable.
In each recursion, we take an item $x_i \in S$ whose index is greater than the last item taken to avoid duplicated searches.
However, if $L(X,S)$ has no correction rules, an exponential loss can occur in the computation time with respect to $|S|$.

We solve the issue by monotonicity of confidence values on equivalent lattices described by the following lemma.
\begin{lem}
On any equivalent lattice $L(X,S)$, for any pair of itemsets $Y,Z \in L(X,S)$, if $Y \subseteq Z$ holds, then $\mathrm{conf}(Y \rightarrow \delta^*(Y)) \leq \mathrm{conf}(Z \rightarrow \delta^*(Z))$ holds where $\delta^*(Y)$ and $\delta^*(Z)$ are optimized correction amounts.
\end{lem}
\begin{proof}
If $Y \subseteq Z$ holds, we have $D^\top(Y) \supseteq D^\top(Z)$, and it indicates that $C^\bot_{Y \rightarrow \delta} \supseteq C^\bot_{Z \rightarrow \delta}$ holds for any $\delta$.
From the property of the equivalent lattice, we have $D^\bot(Y) = D^\bot(Z)$, and it indicates that $C^\top_{Y \rightarrow \delta} = C^\top_{Z \rightarrow \delta}$ holds for any $\delta$.
Hence, by the definition of the confidence, we have $\mathrm{conf}(Y \rightarrow \delta) \leq \mathrm{conf}(Z \rightarrow \delta)$ for any $\delta$.
Therefore, we have $\mathrm{conf}(Y \rightarrow \delta^*(Y)) \leq \mathrm{conf}(Z \rightarrow \delta^*(Y)) \leq \mathrm{conf}(Z \rightarrow \delta^*(Z))$.
\end{proof}
\noindent
Given an equivalent lattice $L(X,S)$, when $\mathrm{conf}(X \cup S \rightarrow \delta^*(X \cup S)) < \lambda$ holds, we can skip scanning the lattice because the lemma above indicates any itemset $Y \in L(X,S)$ satisfies $\mathrm{conf}(Y \rightarrow \delta^*(Y)) < \lambda$.
This is a reasonable pruning strategy for each lattice unit that requires only one optimization of the correction amount.

\subsection{Minimal Enumeration}
Output sizes of rule mining tend to be large in general.
Therefore, we propose an option to reduce output sizes of correction rule mining by defining the minimality as follows:
\begin{dfn}[Minimal Correction Rule]
An acceptable correction rule $Y \rightarrow \delta^*(Y)$ is {\it minimal} if there are no acceptable correction rules such that $Z \rightarrow \delta^*(Z)$ where $Y \supset Z$ and $D^\bot(Y) = D^\bot(Z)$.
\end{dfn}
\noindent
Namely, if many rules hit the same false instances, minimality focuses on the rules with minimal itemsets.

We can simply reduce the number of output correction rules by enumerating only minimal ones.
Moreover, we can make the backtrack strategy on equivalent lattices more efficient: if a current itemset becomes an acceptable correction rule, we will immediately backtrack because the after recursion makes only non-minimal correction rules.
After the backtrack search on all the equivalent lattices, we eliminate non-minimal correction rules not detected in the backtrack search.

\subsection{Implementation}
We conclude the pseudo code of the proposed algorithm named CRMiner in Algorithm \ref{alg:crm}.
First, the algorithm uses LCM on $D^\bot$ and threshold $\theta$ to enumerate equivalent lattices $\mathcal{L}$ where each $L(X,S) \in \mathcal{L}$ satisfies $|X| \leq K$.
Second, the algorithm checks the pruning condition of each lattice and scans each non-pruned lattice by the backtrack search.
For the details of LCM, please refer to the literature~\cite{LCM}.
In particular, we follow LCM in the data structure of storing datasets.

As a preprocessing of CRMiner, we recommend reordering items $I = \{x_1, x_2, \ldots, x_n\}$ by ascending order of the number of instances including them.
According to the LCM literature, item reordering reduces the computation time of managing the instance set hitting a current itemset in a backtrack search.
This is useful for not only LCM but also scanning lattices.

We summarize the computational complexity of CRMiner.
The computation time of managing hit instances is the same as in the LCM literature.
We focus our analysis on optimizing correction amounts and extracting minimal rules.
The process of optimizing correction amounts is shown in Section \ref{subsec:optimize-correction-amount}.
It is done in the linear time of the number of hit instances, previously sorting the instances by prediction scores.
For extracting minimal rules, we construct process units $\mathcal{R}_{\mathrm{key}} := \{X \rightarrow \delta \in \mathcal{R} \mid D^\bot(X) = \mathrm{key}\}$ by hashing while scanning lattices.
This hashing process does not affect the computational complexity of the scanning process.
After that, we extract minimal ones for each unit.
It takes $O(K \cdot |\mathcal{R}_{\mathrm{key}}|^2)$ because it is done by comparing all pairs of itemsets.

\begin{algorithm}[t]
\caption{CRMiner}
\label{alg:crm}
\textbf{Input}: Itemset $I$, datasets $D^\top$ and $D^\bot$, length limit $K$, support threshold $\theta$, and confidence threshold $\lambda$\\
\textbf{Output}: All the acceptable correction rules.
\begin{algorithmic}[1]
\Function{SCAN-LATTICE}{$X, S, i$}
    \State{$\mathcal{R} \leftarrow \emptyset$}
    \If{$X \neq \emptyset$ and $X \rightarrow \delta^*(X)$ is acceptable}
        \State{$\mathcal{R} \leftarrow \mathcal{R} \cup \{X \rightarrow \delta^*(X)\}$}
        \If{Using minimal enumeration}
            \Return $\mathcal{R}$
        \EndIf
    \EndIf
    \If{$|X| = K$}
        \Return $\mathcal{R}$
    \EndIf
    \For{$x_j \in S$ where $j > i$}
        \State{$\mathcal{R} \leftarrow \mathcal{R} \cup \text{SCAN-LATTICE}(X \cup \{x_j\}, S, j)$}
    \EndFor
    \Return $\mathcal{R}$
\EndFunction
\Procedure{CRMiner}{}
    \State{Compute equivalent lattices $\mathcal{L}$ by LCM}
    \State{$\mathcal{R} \leftarrow \emptyset$}
    \For{$L(X,S) \in \mathcal{L}$}
        \If{$\mathrm{conf}(X \cup S \rightarrow \delta^*(X \cup S)) \geq \lambda$}
            \State{$\mathcal{R} \leftarrow \mathcal{R} \cup \text{SCAN-LATTICE}(X, S, 0)$}
        \EndIf
    \EndFor
    \If{Using minimal enumeration}
        \State{Eliminate non-minimal ones on $\mathcal{R}$}
    \EndIf
    \State{Output $\mathcal{R}$}
\EndProcedure
\end{algorithmic}
\end{algorithm}
\subsection{Simple Denoising}
In general, because datasets may have limited size and some biases, enumerated correction rules will contain noise rules with low generalization performances.
Therefore, we introduce a simple denoising method.
We split the given dataset into the {\it mining} and {\it validation} datasets.
We obtain the correction rules by CRMiner running on the mining dataset.
After that, we extract only efficient correction rules according to the validation dataset.
If a correction rule has confidence of less than a given threshold $\lambda_{\mathrm{valid}} \in [0,1]$ on the validation dataset, we regard the correction rule as not efficient and eliminate it.

\section{Applications}

\subsection{Constructing Correction Models}
Given a set of correction rules $\mathcal{R}$, we consider how to directly correct base models.
We define rule list and rule set models as examples of {\it correction models}.
They differ in a way whether the predictions are changed by at most one rule or the average of multiple rules.
\begin{dfn}[Correction Rule List (CRL)]
CRL is an ordered list $(X_1 \rightarrow \delta_1, X_2 \rightarrow \delta_2, \ldots, X_m \rightarrow \delta_m) \in \mathcal{R}^m$.
Given an instance of an itemset $t \subseteq I$ and a prediction score $s \in (-1,1)$, CRL finds the minimum index $i$ of $X_i \subseteq t$.
If such $i$ exists, CRL outputs $s + \delta_i$ as the changed prediction score, otherwise outputs the original $s$.
\end{dfn}
\begin{dfn}[Correction Rule Set (CRS)]
CRS is a subset $\mathcal{S} \subseteq \mathcal{R}$.
Given an instance of an itemset $t \subseteq I$ and a prediction score $s \in (-1,1)$, CRS finds the subset $\mathcal{T} := \{X \rightarrow \delta \in \mathcal{S} \mid X \subseteq t\}$.
If $\mathcal{T} \neq \emptyset$ holds, CRS outputs $s + \frac{1}{|\mathcal{T}|}\sum_{X \rightarrow \delta \in \mathcal{T}}\delta$ as the changed prediction score, otherwise outputs original $s$.
\end{dfn}

A greedy algorithm easily constructs these correction models.
First, we determine the model size limit $M$ if necessary.
Second, we select an objective function such as accuracy, f1 score, log loss, and so on.
After that, we greedily take a correction rule in $\mathcal{R}$ improving the objective value until the size limit is met or the objective value can not be improved.

\subsection{Improving Training Dataset}
We also consider an indirect use case of correction rules.
For making base models more accurate, a general solution is to retrain models with improved training datasets by adding more instances or weighting instances.
Because the concept of the correction rule is a simple prediction change according to an if-then rule and a scalar calculation, we can say that correction rules guess inaccurate subpopulations of probably easy to correct.
Therefore, augmenting instances like those captured by mined correction rules is a good strategy for improving training datasets.

\subsection{Guessing Concept Drift Regions}
Machine learning models often misclassify instances that have not been learned well or generated on unknown distributions (i.e., concept drift).
The two applications described above can simultaneously adapt base models to misclassification cases.
However, some scenarios aim to handle only concept drift.
Therefore, we introduce a method to extract correction rules probably related to concept drift.

We assume that there are two datasets where one is used to train the base model (old dataset) and the other is obtained on a new distribution (new dataset).
For each correction rule mined on the new dataset, if it does NOT adapt to the old dataset, we consider that the correction rule guesses a concept drift region.
Hence, we adopt a correction rule mined on the new dataset if it has confidence values of less than a given threshold $\lambda_{\mathrm{drift}} \in [0,1]$ on the old dataset.

\section{Experiments}
We demonstrate the power of enumerating correction rules compared with other rule-based model corrections.
The experiments were conducted from the following points of view.
\begin{itemize}
    \item Can we efficiently improve model accuracy by directly adapting correction rules?
    \item Can we collect additional training data according to correction rules to effectively improve model accuracy?
    \item Can we capture various concept drift regions by enumerating correction rules?
\end{itemize}
Therefore, we handled two model correction scenarios of {\it data lacking} and {\it concept drift}.

All the codes for rule mining algorithms were implemented in C++20.
All the codes for dataset settings, classification models, and correction models were implemented in Cython-0.29.34 and Python-3.10.6.
The experiments were conducted on Ubuntu 22.04.2 LTS with Intel(R) Xeon(R) CPU E5-2667 v4 @ 3.20GHz and 128GB RAM.
The source codes are available at \url{https://doi.org/10.6084/m9.figshare.24305305.v1}.

\subsection{Setup}

\subsubsection{Scenarios}
We simulated the following two scenarios.

\paragraph*{Scenario 1. Data Lacking}
We constructed a model with slight accuracy trained on a small dataset.
On a later day, we will have a chance to obtain additional data, albeit for a short period.
We are thinking of carefully collecting limited data to improve training data effectively.
What attributes should we focus on to collect data?

\paragraph*{Scenario 2. Concept Drift}
We constructed a better model trained on enough datasets.
However, the model has become inaccurate because of concept drift.
Now, we are trying to effectively update the model by investigating regions related to concept drift.
Can we summarize the concept drift regions?

\subsubsection{Datasets}
We used four datasets Adult~\cite{Adult}, FICO~\cite{FICO}, Magic~\cite{Magic}, and Stability~\cite{Stability}.
Adult is a social dataset to classify individual income exceeding \$50K/year based on census data.
FICO is a finance dataset to classify whether a person has a default risk on a loan through customer information.
Magic is a physical dataset for binary classification of telescope observation patterns.
Stability is a simulated dataset to classify whether an electrical grid system is stable.
We dropped all the instances having a lacking attribute.
A summary of dataset statistics is shown in TABLE \ref{tab:datasets}.

In the data lacking scenario, we split each dataset into training data (TRN) of 100 instances, mining data (MNG) of 500 instances, augmentation data (AUG) of 50\% instances, and test data (TST) of the other instances by the class-wise stratified manner.
Note that, as described later, AUG data were used for sampling instances which were regarded as collected data for model re-training.
In the concept drift scenario, we split each dataset into TRN of 40\% instances, MNG of 40\% instances, and TST of 20\% instances.
The synthetic generation of concept drift is explained later.
For both scenarios, MNG data were discretized by categorical one-hot encoding and numerical quantile cut of four bins.
Items derived from numerical attributes were denoted by ``$c \leq \mathrm{attr}$`` or ``$\mathrm{attr} < c$`` for each cut point $c$.
Moreover, all the experiments were conducted on 10 different dataset splits and summarized their results.

\begin{table}[t]
\centering
\caption{Summary of the real datasets.}
\label{tab:datasets}
\begin{tabular}{c|r|r|r|r}
Name & \#Instance & \#Positive Instance & \#Attribute & \#Item \\\hline
Adult & 45,222 & 11,208 & 14 & 120 \\\hline
FICO & 10,459 & 5,459 & 23 & 124 \\\hline
Magic & 19,020 & 12,332 & 10 & 60 \\\hline
Stability & 10,000 & 3,620 & 11 & 66 \\\hline
\end{tabular}
\end{table}

\subsubsection{Synthetic Generation of Concept Drifts}
Concept drift scenarios were randomly generated for each dataset.
First, we divided the instances into disjoint subsets by a decision tree of depth five over random dummy class labels.
Here, cut points of numerical attributes were previously generated by the quantile cut of four bins.
Second, we extracted 10 subsets each of which satisfied the following two conditions:
(1) The number of instances is between 3\% and 5\% of the original.
(2) Both the proportions of the positive and negative instances are at least 10\% on the subset.
Until we succeeded in strictly extracting such 10 subsets, we repeated the random generation above.
Finally, we randomly selected a class label for each extracted subset.
Then, for each subset, we set the instances of the selected label so that they were less likely to be included in the TRN data.
As a result, there were 10 disjoint regions of concept drift by label shifts between the TRN and the other data.
After sampling TRN data, MNG  and TST data were sampled in a class-wise stratified manner.

\subsubsection{Base Classification Models}
We used four types of classification models K-Nearest Neighbors (KN), Logistic Regression (LR), Random Forest (RF), and Gradient Boosting Trees (GB) as base models.
For each pair of a dataset and model type, we trained the model on the TRN data by scikit-learn~\cite{sklearn} with optimizing some hyper-parameters by Optuna~\cite{Optuna} of 100 trials.
We set the metric of the training as log loss.
In addition, we used SMOTE~\cite{SMOTE} algorithm for oversampling to balance the class distribution.
More details of the hyper-parameters are shown in TABLE \ref{tab:base-model}.

\begin{table}[t]
\centering
\caption{Settings of Base Classification Models.}
\label{tab:base-model}
\begin{tabular}{c|c|c|c|c}
Model & Parameters in scikit-learn\\\hline
KN &
    \begin{tabular}{c}
        $\text{n\_neighbors} \in \{1,2,\ldots,50\}$\\
        $\text{weights} \in \{\text{uniform}, \text{distance}\}$\\
        $\text{metric} \in \{\text{minkowski}, \text{manhattan}\}$ 
    \end{tabular}\\\hline
LR &
    \begin{tabular}{c}
        $\text{C} \in [1e\mathchar`-4,1e4]$ (log scale)\\
        $\text{penalty} \in \{\text{L1}, \text{L2}\}$
    \end{tabular}\\\hline
RF &
    \begin{tabular}{c}
         $\text{max\_depth} \in \{2,3,\ldots,10\}$\\
         $\text{min\_samples\_leaf} \in \{1,2,\ldots,64\}$ 
    \end{tabular}\\\hline
GB &
    \begin{tabular}{c}
         $\text{learning\_rate} \in [1e\mathchar`-4, 1]$ (log scale)\\
         $\text{max\_depth} \in \{2,3,4,5\}$\\
         $\text{min\_samples\_leaf} \in \{1,2,\ldots,64\}$ 
    \end{tabular}\\\hline
\end{tabular}
\end{table}

\begin{figure*}[h]
\centering
\includegraphics[scale=0.265]{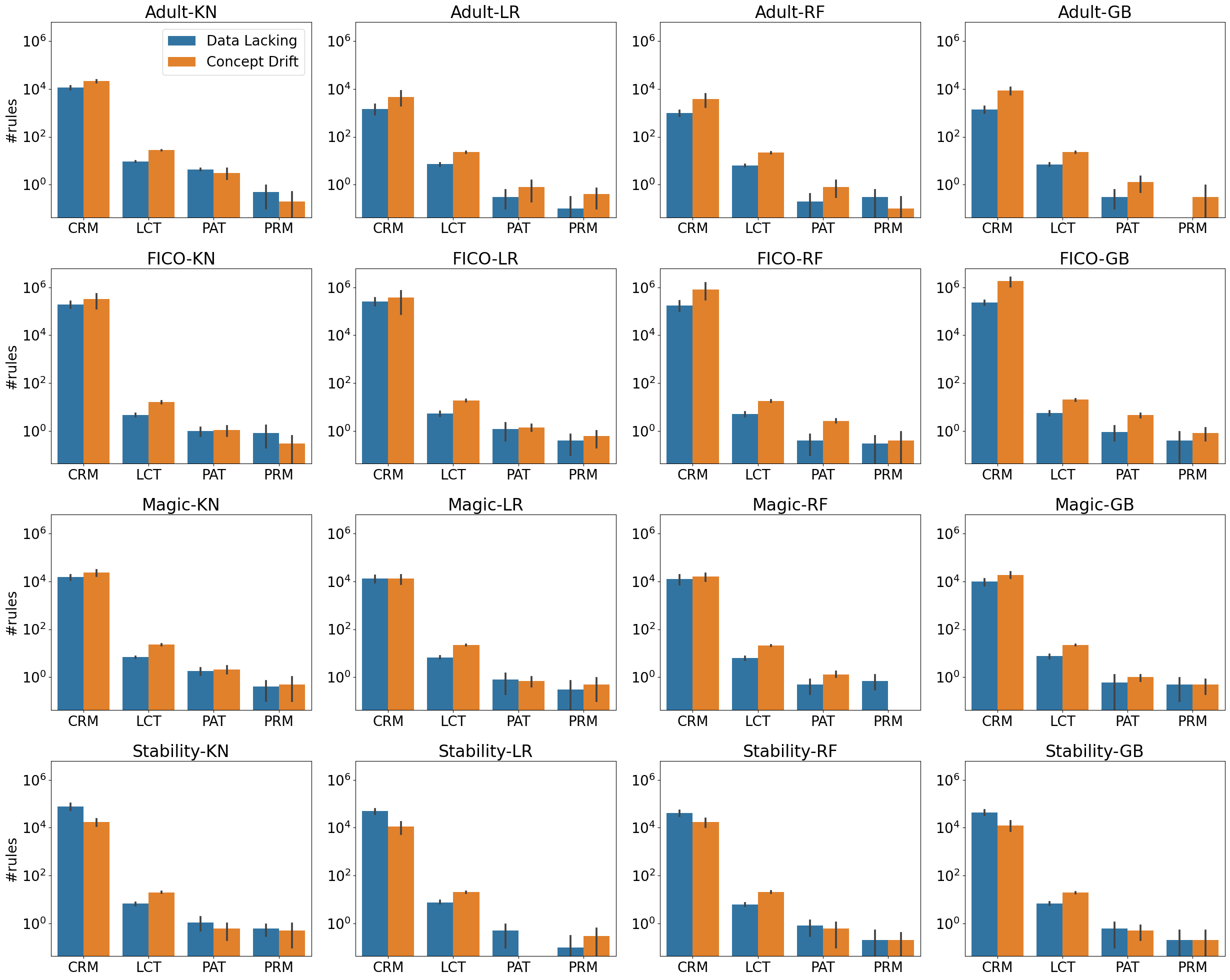}
\vspace{-2mm}
\caption{Number of obtained rules. The black bars are the two-sided 95\% confidence intervals.}
\label{fig:n-rules}
\end{figure*}

\begin{figure*}[h]
\centering
\includegraphics[scale=0.265]{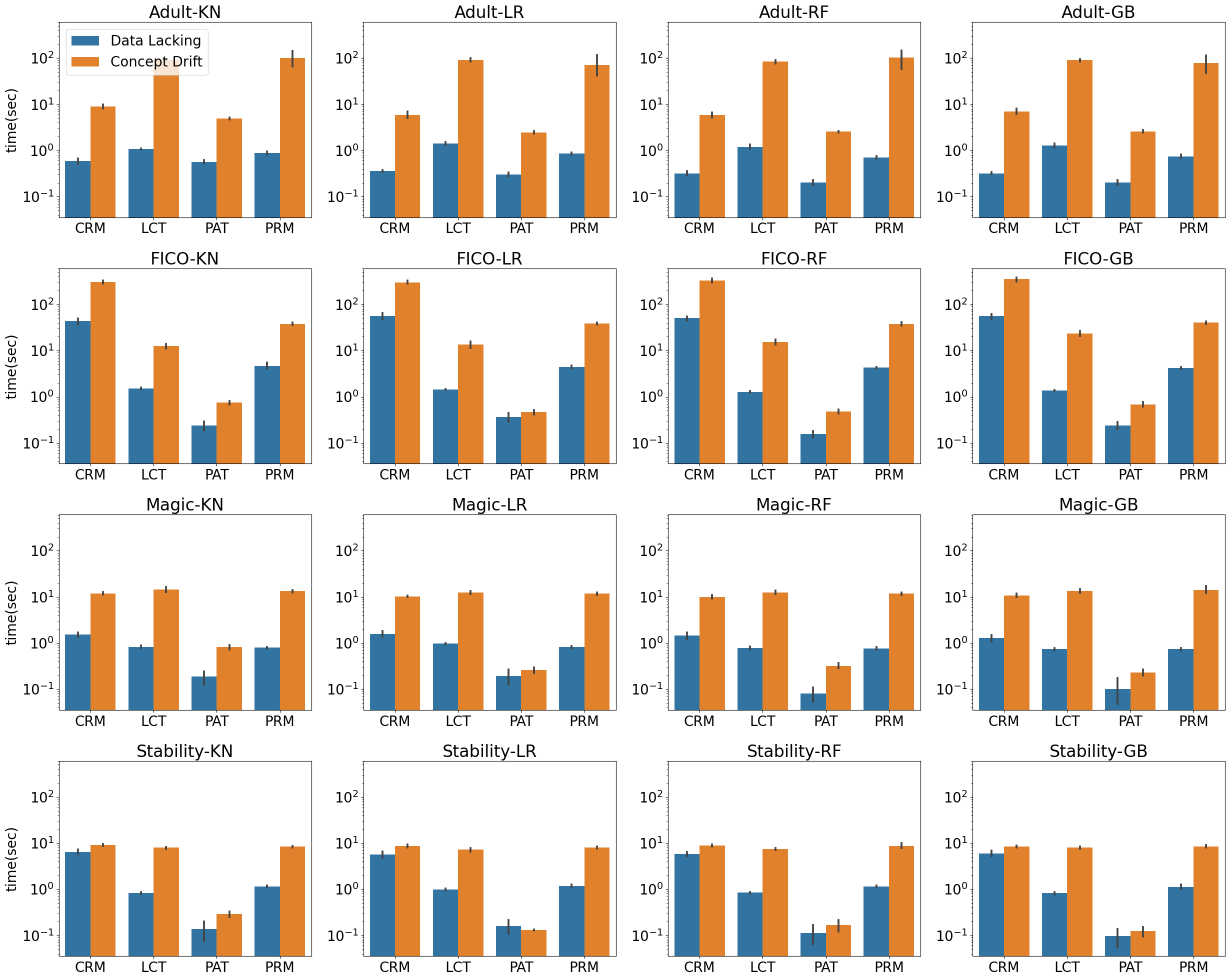}
\vspace{-2mm}
\caption{Number of obtained rules. The black bars are the two-sided 95\% confidence intervals.}
\label{fig:run-time}
\end{figure*}

\subsubsection{Competitors}
We compared the proposed method CRMiner with PREMISE~\cite{PREMISE}, local correction tree (LCT)~\cite{LCT}, and patching~\cite{Patching}.
They run on each pair of MNG data and a base model.
MNG data were split into 80\% and 20\% ratios, with 20\% of data being validation data.
TRN data were used for guessing concept drift regions.
As same as training base models, we used SMOTE algorithm for oversampling to balance the class distribution.
All the methods run with the fixed settings throughout all the experiments as follows:
\begin{itemize}
\item \textbf{CRMiner (CRM)} run with the maximum itemset length $K = 5$, the support threshold $\theta = 0.05$, and the confidence thresholds $\lambda = 0.9$, $\lambda_{\mathrm{valid}} = 0.7$, and $\lambda_{\mathrm{drift}} = 0.5$.
In addition, we used the minimal enumeration strategy.
From these settings, we expect that the obtained correction rules have moderate coverage and high accuracy.
\item \textbf{PREMISE (PRM)} run with the maximum itemset length of five.
For each itemset obtained by PREMISE, an optimized correction amount was assigned for we treated them so that they were correction rules.
On the validation and guessing concept drift regions, we set the confidence threshold $\lambda_{\mathrm{valid}} = 0.7$ and $\lambda_{\mathrm{drift}} = 0.5$.
\item \textbf{LCT} was trained with the depth of five.
On the validation, we eliminated leaves of no accuracy improvement.
On guessing concept drift regions, we ignored leaves with even a slight accuracy improvement on TRN data.
\item \textbf{patching (PAT)} consisted of decision tree detecting misclassification and random forest patch.
A decision tree was trained with the depth of five and the labels of ``misclassified`` or not.
A leaf with a ``misclassified`` label leads to a misclassification region.
For each misclassification region, a random forest was trained as a patch without tuning.
The settings above are imitations of the patching literature~\cite{Patching}.
On the validation, we eliminated patches of no accuracy improvement.
On guessing concept drift regions, we ignored misclassification regions with even a slight accuracy improvement on TRN data.
\end{itemize}
\subsection{Direct Correction of Classification Models}
We adapted each method to directly correct the base models and measured accuracy changes.
For both CRM and PRM, we constructed an obvious correction rule set (CSR) having all the mined correction rules.
There was no optimization on CRS to observe averaged behaviors of mined correction rules.
LCT and PAT are already correction models and need no additional process for direct correction.
We summarized the results on 10 different TRN-MNG-(AUG)-TST splits.

Before showing the main results, we check the number of obtained rules for each method as shown in Fig. \ref{fig:n-rules}.
Surprisingly, in almost all the cases, CRM obtained over 100 times the number of rules than the other methods, although CRM aimed to find highly confident rules.
This indicates that the non-enumeration methods may miss a lot of confident rules.
On the other hand, the large outputs of CRM need to be properly reduced in practical human checking situations.
We conducted a reduction of CRM outputs in the experiments of the data collection and the concept drift summarization.

We also check the run times for each method as shown in Fig. \ref{fig:run-time}.
The run times of LCT, PAT, and PRM were relatively longer for the most large dataset Adult.
In fact, the data size is the most significant and empirical factor of their run times.
Especially, PAT can be exactly completed in polynomial time relative to the data size and is exceptionally fast among the competitors.
Looking back at Fig. \ref{fig:n-rules}, CRM seems to take longer run times by obtaining more rules.
This probably indicates a property that CRM appropriately prunes the redundant itemsets and reduces run times on exponentially huge search space.
In particular, CRM was only equal to or 10 times slower than LCT and PRM in run times, despite obtaining over 100 times larger number of rules.

The results of the accuracy improvements on data lacking and concept drift cases are shown in Fig. \ref{fig:lack} and Fig. \ref{fig:drift}, respectively.
We focused on four metrics accuracy (acc), precision (pre), recall (rec), and f1 score (f1).
The grayscale heat maps show the averaged rankings of accuracy changes among four competitors on TST data (darker is better).
The red-blue heat maps show the averaged accuracy changes on TST data (red means improving and blue means worsening).
The rows correspond to the competitors and the columns correspond to the base model types.
The values are compared vertically.

As shown in Fig. \ref{fig:lack}, CRM and LCT were superior to PAT and PRM in almost all the data lacking cases.
On Magic dataset, we can clearly find worsening results of a few percent recall down regardless of the methods.
However, CRM obtained such results only for the RF-type base models.
From the above observations, rule-based correction models should be optimized with clear metrics for validating accuracy improvements, as in CRM and LCT.
Furthermore, by focusing on the results of ranking, CRM tends to be better than LCT.
Therefore, it is better to enumerate correction rules than to obtain limited variations of rules, e.g., by a tree, because data lacking cases can make it difficult to identify appropriate regions where correction is effective.

As shown in Fig. \ref{fig:drift}, CRM and LCT were superior to PAT and PRM in almost all the concept drift cases.
Fortunately, there were no clear worsening results.
Focusing on the results of ranking, LCT tends to be better than CRM.
Therefore, in cases with obvious degradation of model performances such as concept drift, limited disjoint rules are possibly better than enumerated rules with some overlap regions.
Note that we did not optimize the correction models derived from CRM.
Hence, by optimizing correction models, more generalization performances could be achieved via CRM.
It is one of future work to develop a method for efficiently optimizing correction models via CRM.

\subsection{Model Re-training with Improving Training Dataset}
We adapted the obtained rules to collect additional data for re-training the base models.
We assumed that MNG data of 500 instances was the first randomly collected data and additional 500 instances would be carefully collected.
Thus, the new training dataset consisted of TRN, MNG, and additional data of 500 instances sampled from AUG data.
The sampling was conducted as follows.
First, a fully random sampling approach (RND) was used as a baseline method.
Second, for each comparison method, we randomly selected 10 rules obtained by the method.
After that, we randomly sampled 50 instances hit by the rule on AUG data for each selected rule.
If the method had less than 10 rules, the lacking instances were made up by uniform random sampling, i.e., the sampling was close to RND approach.
Here, because CRM generated more rules than other methods, we reduced the target rules by greedily constructing a correction rule list of 16 rules by focusing on accuracy improvement.
We summarized the results on 10 different TRN-MNG-AUG-TST splits.

The results are shown in Fig. \ref{fig:re-train} (the composition of the figure is the same as in Fig. \ref{fig:lack} and \ref{fig:drift}).
CRM aimed to obtain highly accurate rules as mentioned in the setup details.
Thus, the most notable metric is accuracy.
In fact, CRM achieved high accuracy improvements for Adult, Magic, and Stability.
On the other hand, CRM could not achieve high recall improvements.
Therefore, when we focus on a specific metric, it may be useful to adjust mining parameters or limit the direction of corrections.
On FICO dataset, we observed that model performances were more improved when the method was closer to RND approach.
Note that, because PAT and PRM obtained only a few rules, their behaviors were close to that of RND in our settings.
That is to say, uniform data collections are better for FICO dataset.
We are guessing the reason as, because misclassifications were found in every possible region on FICO dataset, and we achieved little accuracy improvements by focusing on a limited part of the misclassification regions.

\begin{figure*}[p]
\centering
\includegraphics[scale=0.23]{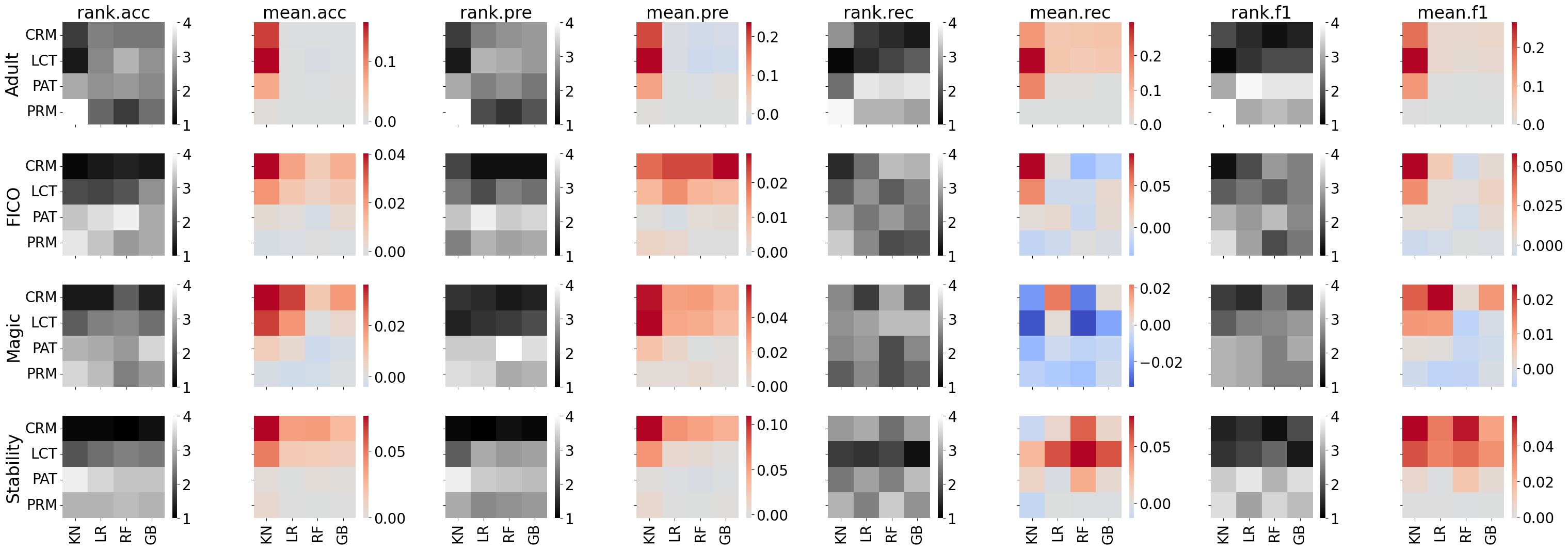}
\vspace{-2mm}
\caption{Summary of accuracy changes by direct model correction in the cases of data lacking.}
\label{fig:lack}
\vspace{4mm}
\centering
\includegraphics[scale=0.23]{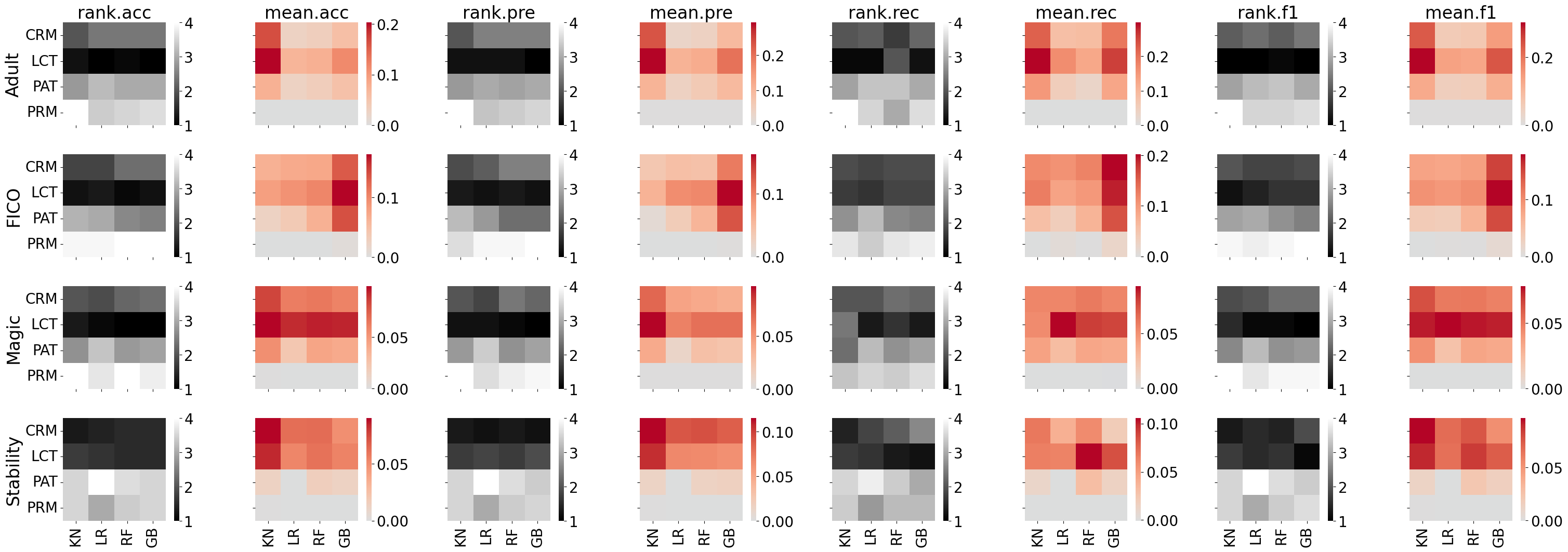}
\vspace{-2mm}
\caption{Summary of accuracy changes by direct model correction in the cases of concept drift.}
\label{fig:drift}
\vspace{4mm}
\centering
\includegraphics[scale=0.325]{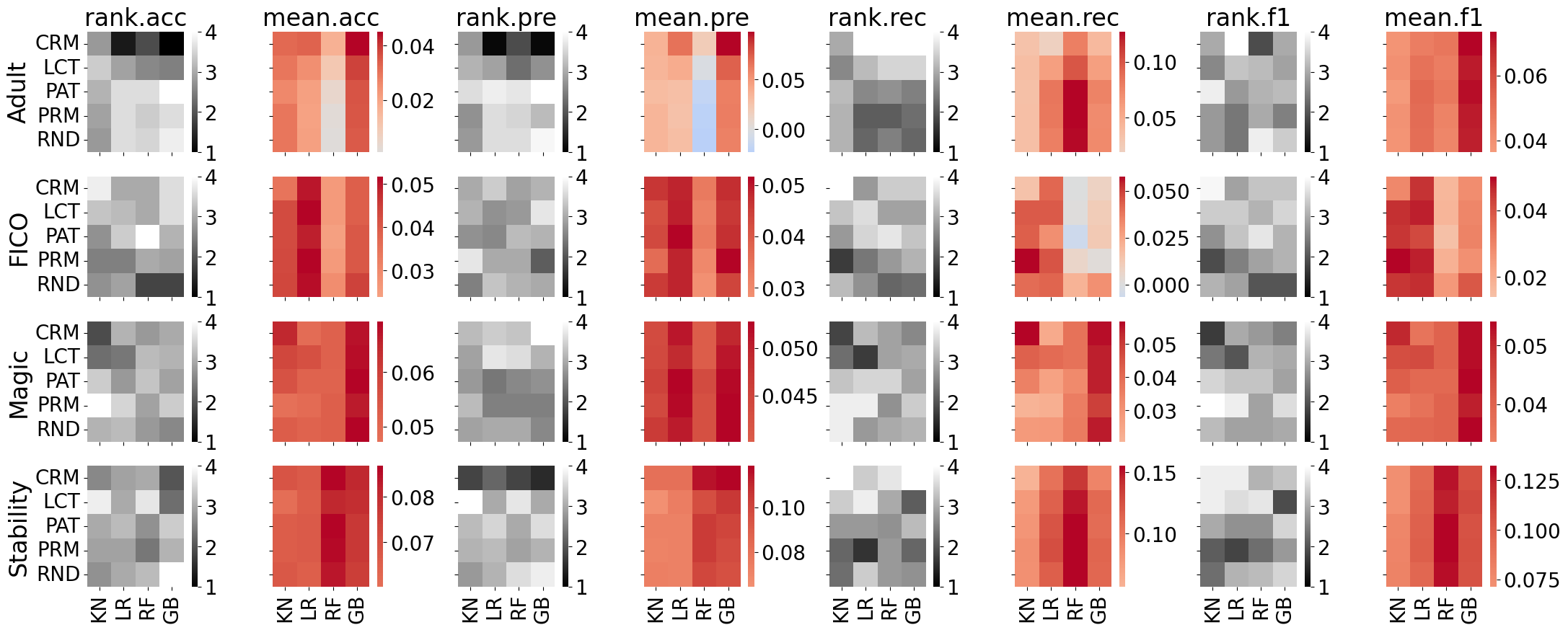}
\caption{Summary of accuracy changes by re-training with additional collected data in the cases of data lacking.}
\label{fig:re-train}
\vspace{5mm}
\end{figure*}
\begin{figure*}[h]
\centering
\includegraphics[scale=0.3]{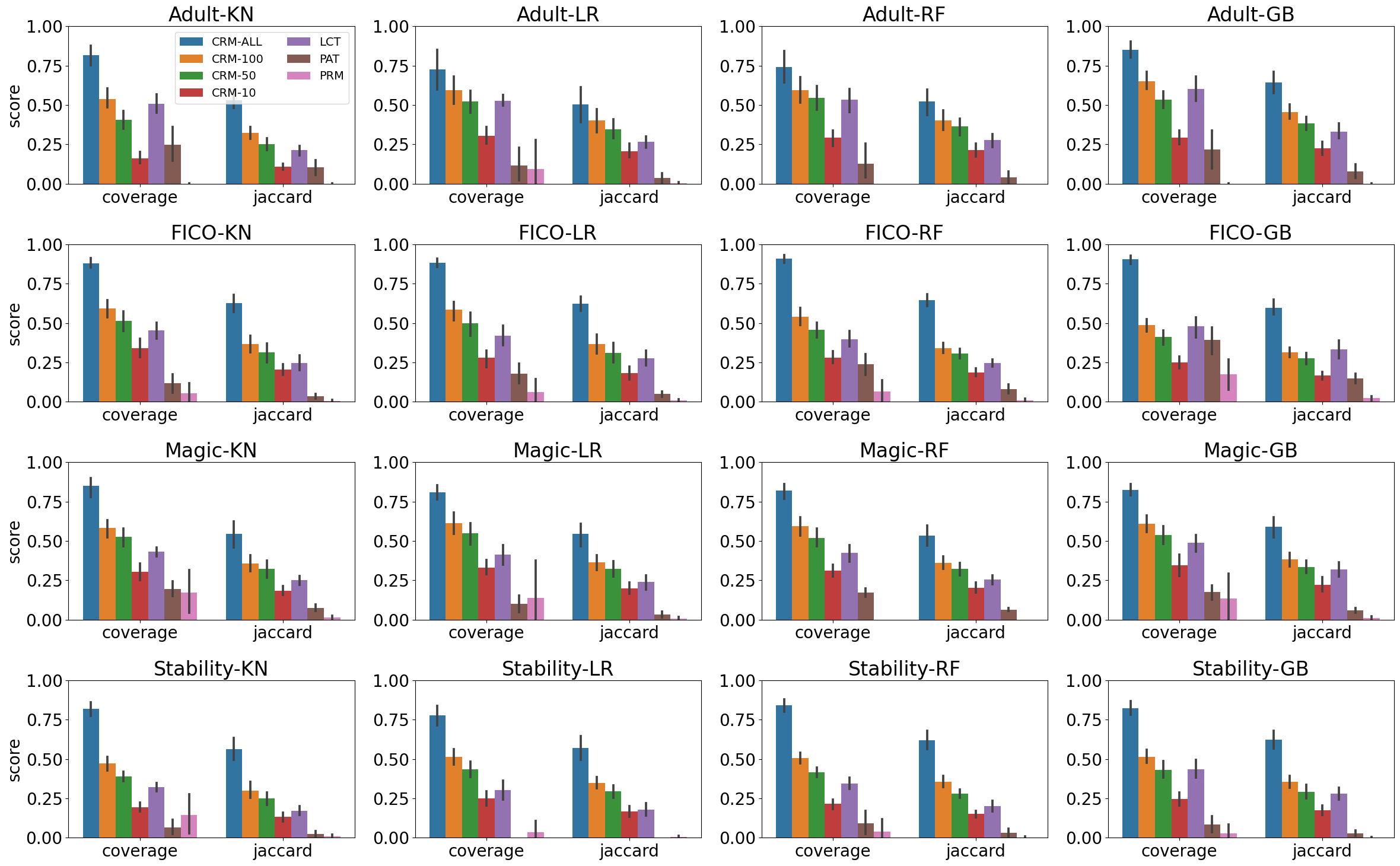}
\caption{Coverage and Jaccard similarity against ground truth concept drift regions. The black bars are the two-sided 95\% confidence intervals.}
\label{fig:coverage}
\end{figure*}
\subsection{Summarization of Concept Drifts}
We analyzed the obtained rules in the concept drift cases.
As described in the setup section, we generated the 10 ground truth rules of concept drift.
We then computed how well each method restored the ground-truth rules in terms of instance subsets hit by rules.
Let $\mathcal{G}$ and $\mathcal{R}$ be the sets of the ground truth rules and the mined rules by a comparison method, respectively.
For each $G \in \mathcal{G}$, we computed the coverage and Jaccard similarity on TST data as follows:
\begin{align*}
\mathrm{coverage}(G,\mathcal{R}) &= \max_{R \in \mathcal{R}} \frac{|D_{\mathrm{TST}}(G) \cap D_{\mathrm{TST}}(R)|}{|D_{\mathrm{TST}}(G)|}\\
\mathrm{jaccard}(G,\mathcal{R}) &= \max_{R \in \mathcal{R}} \frac{|D_{\mathrm{TST}}(G) \cap D_{\mathrm{TST}}(R)|}{|D_{\mathrm{TST}}(G) \cup D_{\mathrm{TST}}(R)|}
\end{align*}
where $D_{\mathrm{TST}}(X)$ is the set of hit instances by a rule $X$ on TST data.
We summarized the averaged coverage and Jaccard similarity for 10 ground truth regions over each of 10 different TRN-MNG-TST splits.

Furthermore, because CRM generated more rules than other methods, we conducted clustering and picked some representative rules.
This helps in practical rule analysis by humans.
The clustering was based on the $k$-modes algorithm which adapts the $k$-means algorithm to categorical variables (in our case, each variable was binary of hit on an instance or not).
In addition, the clustering was adapted in a direction-wise manner (each of positive and negative).
After clustering by $k$-modes, we picked the nearest rule for each centroid.
We then obtained up to $2k$ correction rules as representatives.
We tried three types of the number of clusters $k \in \{10, 50, 100\}$.

The results are shown in Fig. \ref{fig:coverage}.
We denote the original CRM and it with $k$-modes as CRM-ALL and CRM-$k$, respectively.
Surprisingly, CRM-ALL obtained the best and enough results for both the coverage and Jaccard similarity in all the cases.
LCT had relatively good results but not as well as CRM.
Moreover, PAT and PRM completely failed to restore the ground truth rules as shown by Jaccard similarity.
Those results witness both the significance of the enumeration and the poorness of non-enumeration.
By focusing on the results of $k$-modes, although CRM-10 was insufficient, increasing $k$ efficiently improved the results.
Especially, CRM-50 was sufficient for being superior to LCT in almost all cases.
Therefore, clustering is a better approach to summarizing correction rules obtained by CRM into a form that is easier for human checks.
It is one of future work to find more effective clustering or other summarization approaches to reduce the burden of a human being.

\section{Conclusion}

In this study, we defined correction rules that describe inaccurate subpopulations and how to correct them.
We developed an efficient mining algorithm to enumerate correction rules.
Our mining algorithm is an efficient backtrack search using the monotonicity on equivalent lattice defined on frequent itemsets.
We can use output rules to understand inaccurate subpopulations and improve the accuracy of prediction models. 
We also proposed various options and applications such as reducing output rules by minimal enumeration, eliminating noisy rules, constructing correction models, collecting additional training data, and summarizing concept drift.
Our experimental results demonstrated the effectiveness of our approach in terms of improved accuracy and concept drift summarization compared to existing rule-based model correction approaches.

\section*{Acknowledgement}
Satoshi Hara is partially supported by JSPS KAKENHI Grant Number 20K19860.

\bibliographystyle{IEEEtran}
\bibliography{icdm23}

\end{document}

%% file: macro.tex
\newtheorem{dfn}{Definition}
\newtheorem{lem}{Lemma}
\newtheorem{prop}{Proposition}

\newcommand{\argmax}{\mathop{\rm arg~max}\limits}
\newcommand{\argmin}{\mathop{\rm arg~min}\limits}

\newcommand{\TODO}[1]{\textcolor{magenta}{#1}}